\documentclass[letterpaper,10pt,english]{article}

\usepackage[T1]{fontenc}
\usepackage{babel}
\usepackage{hyperref}

\usepackage{authblk}
\usepackage{amsmath}
\usepackage{amssymb}
\usepackage{amsthm}
\usepackage{tikz}
\usepackage{graphicx}
\usepackage{float}
\usepackage[ruled,noend]{algorithm2e}
\usepackage[noend]{algpseudocode}
\usepackage{tabularx}
\usepackage[margin=1in]{geometry}
\usepackage{arydshln}
\usepackage{multirow}
\usepackage{wrapfig}
\usepackage{url}

\usetikzlibrary{positioning,calc,shapes,arrows}
\usetikzlibrary{backgrounds}
\usetikzlibrary{matrix,shadows,arrows} 
\usetikzlibrary{decorations.pathreplacing}

\def\etal.{et\penalty50\ al.}
\theoremstyle{plain}
\newtheorem{theorem}{Theorem}[section]
\newtheorem{lemma}[theorem]{Lemma}

\newtheorem{corollary}[theorem]{Corollary}

\theoremstyle{definition}
\newtheorem{definition}{Definition}[section]

\theoremstyle{remark}
\newtheorem{question}{Question}[section]

\newtheorem{observation}[question]{Observation}

\theoremstyle{plain}
\newtheorem*{theorem*}{Theorem}

\providecommand{\keywords}[1]
{
  \small	
  \textbf{\textit{Keywords---}} #1
}

\DeclareMathOperator{\ttime}{ttime} 

\author{Hovhannes A. Harutyunyan}
\author{Kamran Koupayi} 
\author{Denis Pankratov}

\affil{Department of Computer Science and Software Engineering\\ Concordia University\\ Montreal, Canada\\ \url{haruty@cs.concordia.ca}, \url{kamran.koupayi@gmail.com}, \url{denis.pankratov@concordia.ca}}

\title{Temporal Separators with Deadlines\thanks{This research is supported by NSERC, Canada.}} 
\date{}

\begin{document}

\maketitle

\begin{abstract}
  We study temporal analogues of the Unrestricted Vertex Separator problem from the static world. An $(s,z)$-temporal separator is a set of vertices whose removal disconnects vertex $s$ from vertex $z$ for every time step in a temporal graph. The $(s,z)$-Temporal Separator problem asks to find the minimum size of an $(s,z)$-temporal separator for the given temporal graph. The $(s,z)$-Temporal Separator problem is known to be $\mathcal{NP}$-hard in general, although some special cases (such as bounded treewidth) admit efficient algorithms \cite{fluschnik2020temporal}.

  We introduce a generalization of this problem called the $(s,z,t)$-Temporal Separator problem, where the goal is to find a smallest subset of vertices whose removal eliminates all  temporal paths from $s$ to $z$ which take less than $t$ time steps. Let $\tau$ denote the number of time steps over which the temporal graph is defined (we consider discrete time steps).  We characterize the set of parameters $\tau$ and $t$ when the problem is $\mathcal{NP}$-hard and when it is polynomial time solvable. Then we present a $\tau$-approximation algorithm for the $(s,z)$-Temporal Separator problem and convert it to a $\tau^2$-approximation algorithm for the $(s,z,t)$-Temporal Separator problem. We also present an inapproximability lower bound of $\Omega(\ln(n) + \ln(\tau))$  for the $(s,z,t)$-Temporal Separator problem assuming that $\mathcal{NP}\not\subset\mbox{\sc Dtime}(n^{\log\log n})$. Then we consider three special families of graphs: (1) graphs of branchwidth at most $2$, (2) graphs $G$ such that the removal of $s$ and $z$ leaves a tree, and (3) graphs of bounded pathwidth. We present polynomial-time algorithms to find a minimum $(s,z,t)$-temporal separator for (1) and (2). As for (3), we show a polynomial-time reduction from the Discrete Segment Covering problem with bounded-length segments to the $(s,z,t)$-Temporal Separator problem where the temporal graph has bounded pathwidth. 

\end{abstract}

\keywords{Temporal graphs, dynamic graphs, vertex separator, vertex cut, separating set, deadlines, inapproximability, approximation algorithms} 


\section{Introduction}
\label{sec:intro}

Suppose that you have been given the task of deciding how robust a train system of a given city is with respect to station closures. For instance, is it possible to disconnect the two most visited places, e.g., the downtown and the beach, by shutting down $5$ train stations in the city? Does an efficient algorithm even exist? If not, what can we say about special classes of graphs? These are central questions of interest in this work.

More formally, we model the scenario as a graph problem. An important component missing from the classical graph theory is the ability of the graph to vary with time. The trains run on a schedule (or at least they are supposed to -- for simplicity, we assume a perfectly punctual train system). Thus, it is not accurate to say that there is an edge between station $A$ and station $B$ just because there are tracks connecting them. It would be more accurate to say that if you arrive at  $A$ at some specific time $t$ then you could get to  $B$ at some other time $t' > t$, where $t$ is when the train arrives at station $A$ and $t'$ is the time when this train reaches station $B$. In other words, we can consider the edge from $A$ to $B$ as being present at a particular time (or times) and absent otherwise. This is an important point for the robustness of train networks, since it could be that due to incompatibility of certain train schedules the train network could become disconnected by shutting down even fewer stations than we otherwise would have thought if we didn't take time schedules into account.

The notion of graphs evolving with time has several formal models in the research literature \cite{anagnostopoulos2012algorithms, rossi2013modeling}. First of all, there is an area of online algorithms~\cite{albers2003online} where the graph is revealed piece by piece (thus the only allowable changes are to add objects or relations to the graph) and we need to make irrevocable decisions towards some optimization goal as the graph is being revealed. Secondly, streaming and semi-streaming graph algorithms deal with graphs that are revealed one piece at a time similar to online algorithms, but the emphasis is on memory-limited algorithms \cite{feigenbaum2005graph, feigenbaum2004graph}. Thus, in streaming one does not have to make irrevocable decisions, but instead tries to minimize the memory size necessary to answer some queries at the end of the stream. Thirdly, there is a notion of dynamic graph algorithms where the emphasis is on designing efficient data structures to support certain queries when the graph is updated by either adding or removing vertices or edges \cite{sleator1983data}. The goal is to maintain the data structures and answer queries, such as ``are nodes $u$ and $v$ connected?'', in the presence of changes more efficiently than recomputing the answer from scratch on every query. It is evident that none of these models is a good fit for our question: the train system is known in advance and it is not frequently updated (some cities that shall remain unnamed take decades to add a single station to the system). Fortunately, there is yet another model of graphs changing with time that has recently gotten a lot of attention and it happens to capture our situation perfectly. The model is called a temporal graph. In this work, we focus on undirected temporal graphs that have a fixed node set but whose edge sets change in discrete time units, all of which are known in advance. Other temporal graph models where changes to nodes are allowed and where time is modelled with the continuous real line have been considered in the research literature but they are outside of the scope of this work. We typically use $\tau$ to indicate the total number of time steps over which a given temporal graph is defined. For example, if we model the train system as a temporal graph with one minute-granularity and the schedule repeats every $24$ hours then the temporal graph would have $\tau = (24 H) \times (60 M/H) = 1440 M$ time steps in total. For emphasis, when we need to talk about non-temporal graphs and bring attention to their unchanging nature we shall call them ``static graphs.''

We study temporal analogues of the Unrestricted Vertex Separator problem from the static world. An $(s,z)$-temporal separator is a set of vertices whose removal disconnects vertex $s$ from vertex $z$ for every time step in a temporal graph. The $(s,z)$-Temporal Separator problem asks to find the minimum size of an $(s,z)$-temporal separator for the given temporal graph. The $(s,z)$-Temporal Separator problem is known to be $\mathcal{NP}$-hard in general \cite{zschoche2020complexity}, although some special cases (such as bounded treewidth) admit efficient algorithms \cite{fluschnik2020temporal}. This question can be thought of as a mathematical abstraction of the robustness of the train network of a city question posed at the beginning of this section. The $(s,z)$-Temporal Separator problem asks you to eliminate all temporal paths between $s$ and $z$ by removing some nodes. Observe that, practically speaking, in real life, one doesn't actually have to eliminate all temporal paths between $s$ and $z$ -- one would have to remove only reasonable temporal paths between $s$ and $z$. Which paths would be considered unreasonable? We consider paths taking too much time as unreasonable. For example, if normally it takes $30$ minutes to get from downtown to the beach, then eliminating all routes that take at most $4$ hours would surely detract most downtown dwellers from visiting the beach. Motivated by such considerations, we introduce a generalization of the $(s,z)$-Temporal Separator problem called $(s,z,t)$-Temporal Separator problem, where the goal is to find the smallest subset of vertices whose removal eliminates all  temporal paths from $s$ to $z$ which takes less than $t$ time steps. Observe that setting $t = \tau$ captures the $(s,z)$-Temporal Separator problem as a special case of the $(s,z,t)$-Temporal Separator problem. Our results can be summarized as follows:
  
   In Section~\ref{sec:general-graphs-hardness}, we present a characterization of parameters $t$ and $\tau$ when the problem is $\mathcal{NP}$-hard. We also present an inapproximability lower bound of $\Omega(\ln(n) + \ln(\tau))$  for the $(s,z,t)$-Temporal Separator problem assuming that $\mathcal{NP}\not\subset\mbox{\sc Dtime}(n^{\log\log n})$. In Section~\ref{section:approximation}, we present a $\tau$-approximation algorithm for the $(s,z)$-Temporal Separator problem, and we convert it to a $\tau^2$-approximation algorithm for $(s,z,t)$-Temporal Separator problem. 
   
   In Section~\ref{sec:bounded_branchwidth}, we present a polynomial-time algorithm to find a minimum $(s,z,t)$-temporal separator on temporal graphs whose underlying graph (see Section~\ref{sec:prelim}) has branchwidth at most $2$. In Section~\ref{sec:tree_like}, we present another polynomial-time algorithm for temporal graphs whose underlying graph becomes a tree after removal of $s$ and $z$. In Section~\ref{sec:disc_reduction}, we show a polynomial-time reduction from the Discrete Segment Covering problem with bounded-length segments to the $(s,z,t)$-Temporal Separator problem where the temporal graph has bounded pathwidth. Therefore, solving the $(s,z,t)$-Temporal Separator problem on a temporal graph whose underlying graph has bounded pathwidth is at least as  difficult as solving the Discrete Segment Covering problem where lengths of all segments are bounded.



\section{Preliminaries}
\label{sec:prelim}

Temporal graphs (also known as dynamic, evolving \cite{ferreira}, or time-varying \cite{flocchini,casteig} graphs) are graphs whose edges are active at certain points in time. A temporal graph $G = (V, E, \tau)$ contains a set of vertices $V$, and a set of edges $E \subseteq V \times V \times [\tau]$ \footnote{Notation $[n]$ stands for $\{1,2,\dots,n\}$.}. So each edge\footnote{We only consider undirected graphs in this work, i.e. no self-loops and $(u,v,t) \in E$ if and only if $(v,u,t) \in E$}. $e \in E$ contains two vertices of $V$ and a time label $t \in [\tau]$ indicating a time step at which the edge is active. A graph $G_{\downarrow} = (V, E')$ where $E'$ contains every edge $e$ that is active at least once in the temporal graph $G$ is called the \emph{underlying graph} (alternatively, the \emph{footprint}) of the temporal graph $G$. A static graph representing active edges for a specific time $t$ is called the layer of the temporal graph at that time and is denoted by $G_t$. Some other ways of modelling temporal graphs could be found in \cite{michail2016introduction}. We refer to $V(G)$ and $E(G)$ as the set of vertices and edges, respectively, of a graph $G$ (either temporal or static). Also for any subset $U \subseteq V(G)$ we refer to the set of all edges in the subgraph induced by $U$ as $E(U)$, and for any node $v \in V$ we use $E(v)$ to denote the set of all edges incident on $v$. We also use $\tau(G)$ to refer to the number $\tau$ of time labels of the temporal graph $G$. 

A temporal path in a temporal graph is a sequence of edges such that (1) it is a valid path in the underlying graph, and (2) the corresponding sequence of times when the edges are active is non-decreasing. Formally, a sequence $P = [(u_1,v_1,t_1), (u_2,v_2,t_2), \ldots, (u_k,v_k,t_k)]$ of edges in a temporal graph $G$ is called an \emph{$(s,z)$-temporal path} if $s = u_1, v_1 = u_2, \dots, v_{k-1} = u_k, v_k = z$ and $t_1 \leq t_2 \leq \dots \leq t_k$. If the sequence of times is in strictly increasing order, the temporal path is called \emph{strict}. \emph{Travelling time} of $P$, denoted by $\ttime(P)$, is defined as $\ttime(P) = t_k-t_1+1$, i.e., the time it takes to travel from $s$ to $z$. If $\ttime(P) \le t$ then we refer to $P$ as an $(s,z,t)$-temporal path. A temporal graph $G$ is \emph{connected} if for any pair of vertices $s,z \in V(G)$ there is at least one temporal path from $s$ to $z$. A temporal graph $G$ is \textit{continuously connected} if for every $i \in [\tau(G)]$ layer $G_i$ is connected. 

We distinguish between three types of temporal paths: (1) \emph{shortest $(s,z)$-temporal path}: a temporal path from $s$ to $z$ that minimizes the number of edges; (2) \emph{fastest $(s,z)$-temporal path}: a temporal path from $s$ to $z$ that minimizes the traveling time; (3) \emph{foremost $(s,z)$-temporal path}: a temporal path from $s$ to $z$ that minimizes the arrival time at destination. \emph{Temporal distance} from node $s$ to node $z$ is equal to the traveling time of the fastest $(s,z)$-temporal path.

A set $S \subseteq V - \{s,z\}$ is called a \emph{(strict) $(s,z)$-temporal separator} if the removal of vertices in set $S$ removes all (strict) temporal paths from $s$ to $z$.  The \emph{(strict) $(s,z)$-Temporal Separator problem} asks to find the minimum size of a (strict) $(s,z)$-temporal separator in a given temporal graph $G$. This problem has been studied before (see Section~\ref{sec:related}). In this work, we propose a new problem that is based on the notion of $(s,z,t)$-temporal paths. We define a set of vertices $S$ to be a \emph{(strict) $(s,z,t)$-temporal separator} if every (strict) $(s,z,t)$-temporal path contains at least one vertex in $S$, i.e., removal of $S$ removes all (strict) $(s,z,t)$-temporal paths. Thus, the new problem, which we refer to as the \emph{(strict) $(s,z,t)$-Temporal Separator problem} is defined as follows: given a temporal graph $G$, a pair of vertices $s,z \in V(G)$, and a positive integer $t$, the goal is to compute the minimum size of a $(s,z,t)$-temporal separator in $G$.

\begin{lemma}
\label{lemma:temporalPath}
Given a temporal graph $G = (V,E,\tau)$ and two distinct vertices $s$ and $z$ as well as an integer $t$, it is decidable in time $O(|S||E|)$ if there is a $(s,z,t)$-temporal path in $G$ where $S = \{t' \;|\; \exists u: (s,u,t')\in E \}$.
\end{lemma}
\begin{proof}
\cite{wu2014path} and \cite{xuan2003computing} present an algorithm that computes fastest paths from a single source $s$ to all of the vertices in $O(|S|(|V| + |E|))$. We could ignore isolated vertices, then we could compute a fastest path from $s$ to $z$ in $G$ and check if its travelling time is at least $t$.
\end{proof}

Branch decomposition and branchwidth of a graph is defined as follows.

\begin{definition}[Branch Decomposition] \textsc{\cite{deng2013multi}}
Given a graph $G = (V, E)$, a branch decomposition is a pair $(T, \beta)$, such that
\begin{itemize}
    \item $T$ is a binary tree with $|E|$ leaves, and every inner node of T has two children.
    \item $\beta$ is a mapping from $V(T)$ to $2^E$ satisfying the following conditions:
    \begin{itemize}
        \item For each leaf $v \in V(T)$, there exists $e \in E(G)$ with $\beta(v) = \{e\}$, and there are no $v, u \in V(T), v \neq u$ such that $\beta(v) = \beta(u)$.
        \item  For every inner node $v \in V (T)$ with children $v_l, v_r, \beta(v) = \beta(v_l)\cup\beta(v_r)$;
    \end{itemize}
\end{itemize}
\end{definition}

\begin{definition}[Boundary] \textsc{\cite{deng2013multi}}
Given a graph $G = (V, E)$, for every set $F \subseteq E$, the boundary $\partial F = \{v | v$ is incident to edges in $F$ and $E \backslash F\}$.
\end{definition}
\begin{definition}[Width of a Branch Decomposition] \textsc{\cite{deng2013multi}}
Given a branch decomposition $(T,  \beta)$ of $G = (V, E)$, the width of this decomposition is $max\{|\partial \beta(v)| \;|\;v \in V (T)\}$.
\end{definition}
The branchwidth $bw(G)$ of $G$ is defined as the minimum width of a branch decomposition of $G$ \textsc{\cite{deng2013multi}}. We note that for any fixed $k$ there is a linear time algorithm to check if a graph has branchwidth $k$, and if so, the algorithm outputs a branch decomposition of minimum width~\cite{bodlaender1997}.

Path decomposition and pathwidth of a graph are defined as follows.

\begin{definition}[Path Decomposition] \cite{ROBERTSON198339}
Given a graph $G = (V, E)$, a path decomposition of $G$ is a pair $(P, \beta)$, such that
\begin{itemize}
    \item $P$ is a path  with nodes $a_1, \ldots a_m$.
    \item $\beta$ is a mapping from $\{a_1, \ldots, a_m\}$ to $2^E$ satisfying the following conditions:
    \begin{itemize}
        \item For $e \in E(G)$ there exists $a_i$ such that vertices of $e$ appear in $\beta(a_i)$.
        \item  For every $v \in V(G)$ the set of $a_i$, such that $v$ appears in $\beta(a_i)$, forms a subpath of $P$.
    \end{itemize}
\end{itemize}
The width of a decomposition $(P, \beta)$ is $\max_{a \in V(P)} |\beta(a)|-1$. The pathwidth of a graph $G$ is the minimum width of a path decomposition of $G$.
\end{definition}


\section{Related Work}
\label{sec:related}

Enright et al. in \cite{enright2018deleting} adopt a simple and natural model for time-varying networks which is given with time-labels on the edges of a graph, while the vertex set remains unchanged. This formalism originates in the foundational work of Kempe et al. \cite{kempe2002connectivity}. There has already been a lot of work on temporal graphs, too much to give a full overview. Thus, in this section, we focus only on the  results most relevant to our work.

The fastest temporal path is computable in polynomial time, see, e.g.~\cite{xuan2003computing, wu2016efficient, wu2014path}. A nice property of the foremost temporal path is that it can be computed efficiently. In particular, there is an algorithm that, given a source node $s \in V$ and a time $t_{start}$, computes for all $w \in V \setminus\{ s \}$ a foremost $(s,w)$-temporal path from the time $t_{start}$ \cite{mertzios}.
The running time of the algorithm is $O(n\tau^3 + |E|)$. 
It is worth mentioning that this algorithm takes as input the whole temporal graph $G$. Such algorithms are known as offline algorithms in contrast to online algorithms in which the temporal graph is revealed on the fly. 
The algorithm is essentially a temporal translation of the breadth-first search (BFS) algorithm (see e.g. \cite{clrs} page 531).

While the Unrestricted Vertex Separator problem is polynomial time solvable in the static graph world (by reducing to the Maximum Flow problem), the analogous problem in the temporal graph world, namely, the $(s,z)$-Temporal Separator problem, was shown to be $\mathcal{NP}$-hard by Kempe et al.~\cite{kempe2002connectivity}. Zschoche et al. \cite{zschoche2020complexity} investigate the $(s,z)$-Temporal Separator and strict $(s,z)$-Temporal Separator problems on different types of temporal graphs. A central contribution in \cite{zschoche2020complexity} is to prove that both $(s,z)-$Temporal Separator and Strict $(s,z)$-Temporal Separator are $\mathcal{NP}$-hard for all $\tau \geq 2$ and $\tau \geq 5$, respectively, strengthening a result by Kempe et al. \cite{kempe2002connectivity} (they show $\mathcal{NP}$-hardness of both variants for all $\tau \geq 12$) \cite{zschoche2020complexity}.

Fluschnik et al. \cite{fluschnik2020temporal} show that $(s,z)$-Temporal Separator remains $\mathcal{NP}$-hard on many restricted temporal graph classes: temporal graphs whose underlying graph falls into a class of graphs containing complete-but-one graphs (that is, complete graphs where exactly one edge is missing), or line graphs, or temporal graphs where each layer contains only one edge. In contrast, the problem is tractable if the underlying graph has bounded treewidth, or if we require each layer to be a unit interval graph and impose suitable restrictions on how the intervals may change over time, or if one layer contains all others (grounded), or 
 if all layers are identical (1-periodic or 0-steady), or 
 if the number of periods is at least the number of vertices. It is not difficult to show that this problem is fixed-parameter tractable when parameterized by $k +l$, where $k$ is the solution size and $l$ is the maximum length of a temporal $(s, z)$-path.


Lastly, we note that the classical Vertex Separator problem from the static world is often stated as asking to find a vertex separator such that after its removal the graph is partitioned into two blocks (one containing $s$ and one containing $z$) of roughly equal size\footnote{That is why earlier we referred to a static world problem of interest as the Unrestricted Vertex Separator problem to emphasize that there is no balancedness requirement.}. This ``balanced'' separator restriction makes the problem $\mathcal{NP}$-hard. The temporal separator problems considered in this work do not have such a restriction, and as discussed they  are hard problems due to the temporal component. There is a lot of research on the Vertex Separator problem, but since our versions do not have this ``balancedness'' restriction, we do not discuss it in detail. An interested reader is referred to~\cite{althoby2019} and references therein.


\section{Temporal Separators with Deadlines on General Graphs}
\label{sec:general-graphs}


\subsection{Hardness of Exact and Approximate Solutions}
\label{sec:general-graphs-hardness}

Zschoche et al. \cite{zschoche2020complexity} show that the $(s,z)$-Temporal Separator problem is $\mathcal{NP}$-hard on a temporal graph $G = (V,E,\tau)$ if $\tau \geq 2$ (and it is in $\mathcal{P}$ if $\tau=1$). So, it is obvious that the  $(s,z,t)$-Temporal Separator problem is $\mathcal{NP}$-hard if $t \geq 2$. In this section we strengthen this result by showing that the problem remains $\mathcal{NP}$-hard even when restricted to inputs with $t = 1$ and $\tau \geq 2$. 

Reduction from the minimum satisfiability problem with non-negative variables to $(s,z,1)$-Temporal Separator could be made by adding a path from $s$ to $z$ in layer $G_i$, which contains all the variables in the $i$-th clause. 
So, $(s,z,1)$-Temporal Separator on temporal graphs with a sufficient number of layers is $\mathcal{NP}$-hard. 
However, it is not easy to establish the complexity of $(s,z,t)$-Temporal Separator on temporal graphs with a small number of layers. Here we aim to show that $(s,z,1)$-Temporal Separator remains $\mathcal{NP}$-hard on a temporal graph $G = (V,E,\tau)$ if $\tau$ is equal to $2$. To do that, we construct a reduction from 
the Node Multiway Cut problem. In this problem, one is given a graph $G = (V,E)$ and a set of terminal vertices $Z = \{z_1,z_2, \dots z_k\}$. A multiway cut $S \in V \backslash Z$ is a set of vertices whose removal from $G$ disconnects all pairs of distinct terminals $z_i$ and $z_j$. The goal is to find a multiway cut of minimum cardinality. The Node Multiway Cut problem is $\mathcal{NP}$-hard for $k \geq 3$ \cite{garg1994multiway}. 

\begin{theorem}
\label{thm:smalltime-nonstrict}
For every $t_0\geq 1$, the $(s,z,t)$-Temporal Separator problem is $\mathcal{NP}$-hard on a  temporal graph $G=(V,E, \tau)$ when restricted to inputs with $t=t_0$ and $\tau \geq 2$.
\end{theorem}
\begin{proof}
For a given graph $H$ and three vertices $z_1$, $z_2$,and $z_3$ we construct a temporal graph $G = (V,E,2)$. Let $V = (V(H) \backslash \{z_1,z_2,z_3\}) \cup \{s,z\}$ and each edge $(u,v)$ in $H$, not incident on $z_1,z_2$ or $z_3$, add two edges $(u,v,1)$ and $(u,v,2)$ to $E$. For each $u$ which is a neighbour of $z_1$ add an edge $(s,u,1)$, and for each $v$ which is a neighbour of $z_2$ or $z_3$ add an edge $(v,z,1)$ to $E$. Finally add $(s,u,2)$ for each neighbour $u$ of $z_2$, as well as $(v,z,2)$ for each neighbour $v$ of $z_3$ to the set of edges. We claim that $S \subseteq V \backslash \{s,z\}$ is a $(s,z,1)$-temporal separator if and only if $S$ is a multiway cut for $H$. \\
$\leftarrow$ Suppose that $S$ is a multiway cut in the graph $H$ and $S$ is not a $(s,z,1)$-temporal separator on the temporal graph $G$. So, there is a $(s,z,1)$-temporal path $P$ with $V(P) \subseteq V \backslash S$. Based on the definition of a $(s,z,1)$-temporal path, either all the edges of the path belong to the layer $G_1$ or all of them belong to the layer $G_2$. Let's consider each case separately.
\begin{itemize}
\item \textbf{Case 1}. All edges of $P$ belong to the layer $G_1$. Suppose that the path $P$ starts with an edge $(s,u,1)$, and ends with an edge $(v,z,1)$. Based on the construction of graph $G$, it is clear that $u$ is a neighbour of $z_1$ and $v$ is a neighbour of $z_2$ or $z_3$. Since all the edges in $G$ that are not incident on $s$ or $z$ also appear in the graph $H$, all the edges except the starting and ending edges in $P$ appear in $H$. Construct a new path $P'$ by replacing  $s$ with $z_1$ and  $z$ with $z_2$ or $z_3$ that is adjacent to $v$. There is no vertex $x \in P$ such that $x \in S$, so all the vertices of $P'$ do not appear in $S$. Then $V({P'}) \subseteq V(H) \backslash S$ and this contradicts the assumption of $S$ being a multiway cut.
\item \textbf{Case 2}. All edges of $P$ belong to the layer $G_2$. Suppose that  $P$ starts with on edge $(s,u,2)$ and ends with on edge $(v,z,2)$, then $u$ is neighbour of  $z_2$ and $v$ is neighbour of  $z_3$. Construct a path $P'$ by replacing  $s$ with $z_2$ and  $z$ with $z_3$. So  $P'$ is a valid path in the graph $H$ from $z_2$ to $z_3$, contradicting the definition of $S$.
\end{itemize}
$\rightarrow$ Suppose that $S$ is a $(s,z,1)$-temporal separator in $G$ and $S$ is not a multiway cut in the graph $H$. So there is a path $P$ between two of the vertices $z_1,z_2,z_3$ in $H$ where  $V(P) \subseteq V(H) \backslash S$. By replacing source vertex $z_1$ or $z_2$ with $s$ and terminal vertex $z_2$ or $z_3$ with $z$ we construct a path $P'$ in which $V(P) \subseteq V \backslash S$. Now consider three cases for path source and terminal of $P$. Since all the edges in $P$ except the first and the last one are not incident on $s$ or $z$, they must appear in both layers $G_1$ and $G_2$. We consider all cases for the start and end vertices of $P$ and derive a contradiction in each case (with the fact that $S$ is a $(s,z,1)$-temporal separator:
\begin{itemize}
\item \textbf{Case 1}. If $P$ is between $z_1$ and $z_2$ then $P'$ lies entirely in layer $G_1$.
\item \textbf{Case 2}. If $P$ is between $z_1$ and $z_3$ then $P'$ lies entirely in layer $G_1$.
\item \textbf{Case 3}. If $P$ is between $z_2$ and $z_3$ then  $P'$ lies entirely in layer $G_2$.
\end{itemize}
\end{proof}






Since Strict $(s,z)$-Temporal Separator is $\mathcal{NP}$-hard on a temporal graph with $\tau \ge 5$ \cite{zschoche2020complexity}, 
it is clear that Strict $(s,z,t)$-Temporal Separator is $\mathcal{NP}$-hard even when restricted to inputs with $t \ge 5$ and $\tau \ge 5$. However,  by a small change to the reduction presented by Zschoche et al. \cite{zschoche2020complexity}, which is inspired by \cite{wu2016efficient}, we can show that Strict $(s,z,t)$-Temporal Separator remains $\mathcal{NP}$-hard even when restricted to inputs with $t = 3$ and $\tau = 4$.

\begin{theorem}
\label{thm:strict-np-hard}
Finding a strict $(s,z,3)$-temporal separator on  a temporal graph $G = (V,E,\tau)$ is $\mathcal{NP}$-hard when restricted to inputs with $\tau = 4$.
\end{theorem}
\begin{proof}
We present a reduction from the vertex cover problem to an instance of Strict $(s,z,3)$-Temporal Separator, which has four layers. Given a graph $H$, we construct a temporal graph $G = (V,E,4)$ as an instance of input for the Strict $(s,z,3)$-Temporal Separator problem.
Let $V = \{s_v,v,z_v | v \in V(H)\} \cup \{s,z\}$ and define $E$ as follows:
\begin{gather*}
    E := \{(s,s_v,2), (s_v,v,3), (v,z,4), (s,v,1), (v,z_v,2), (z_v,z,3), (z_v,z,4) | v \in V(H)\} \cup \\
    \{(s_u,z_v,3), (s_v,z_u,3) | (u,v) \in E(H)\}
\end{gather*}

\begin{figure}
    \centering
    \includegraphics[scale=0.7,angle=0,height=6cm]{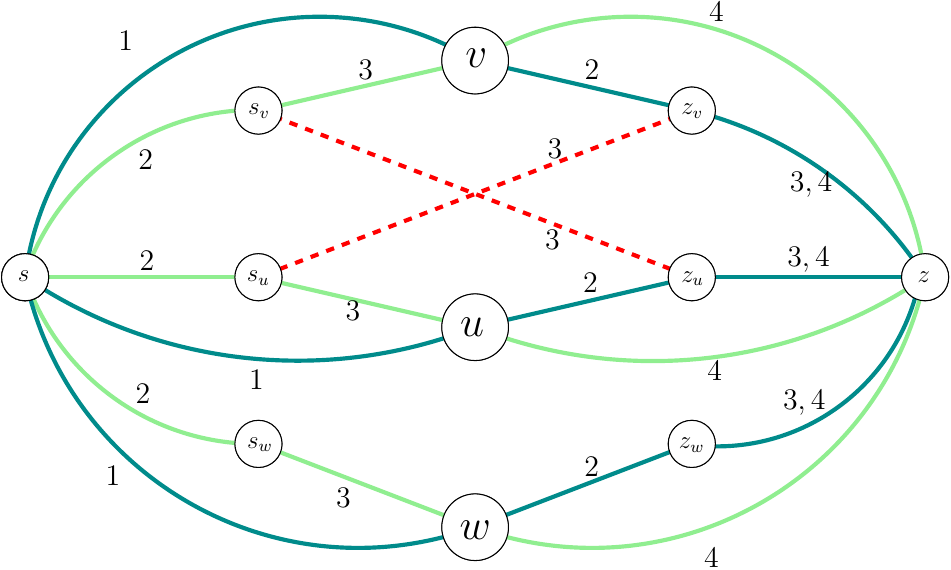}
    \caption{An instance of the Strict $(s,z,3)$-Temporal Separator problem with four layers that corresponds to a vertex cover problem instance in the proof of Theorem~\ref{thm:strict-np-hard}}
    \label{fig:strict-separators}
\end{figure}

Figure \ref{fig:strict-separators} shows the structure of the temporal graph $G$.
Let $n = |V(H)|$; we claim that there is a vertex cover  in $H$ of size $k$, if and only if there exists a strict $(s,z,3)$-temporal separator in $G$ of size $n+k$.

$\rightarrow$ Let $C \in V(H)$ be a vertex cover of size $k$ in $H$, and define $S = \{v | v\in V(H) \backslash C \} \cup \{s_v, z_v | v \in V(H)\}$. Assume that there is a strict $(s,z,3)$-temporal path $P$ such that all its vertices belong to $V \backslash S$. Since for every $v \in V(H)$ either $v \in S$ or $\{s_v,z_v\} \subseteq S$, temporal path $P$ is of the following form:
\begin{gather*}
    P = (s,s_u,2), (s_u,z_v,3), (z_v,z,4).
\end{gather*}
This implies the existence of edge $(s_u, z_v,3)$ in $G$, that results in $(u,v) \in E(H)$. Also existence of $s_u$ and $z_v$ in $P$ implies that $\{u,v\} \subseteq V(H)\backslash C$ which contradicts the fact that $C$ is a vertex cover. So, there is no $(s,z,3)$-temporal path in induced temporal graph $G$ by $V \backslash S$. The cardinality of set $S$ which is a strict $(s,z,3)$-temporal separator for temporal graph $G$ is equal to $(n-k) + 2k$.

$\leftarrow$ Let $S \in V$ be a strict $(s,z,3)$-temporal separator in which $|S| = n + k$. 
For any vertex $v \in V(H)$ we claim that either $v \in S$ or $\{s_v,z_v\} \subseteq S$, otherwise one of the two strict $(s,z,3)$-temporal paths $P_1$ and $P_2$ which are shown in equation \ref{eq:p1} and \ref{eq:p2}, respectively, will not be removed from the graph $G$ by removing  $S$.
\begin{gather}
    \label{eq:p1}P_1 = (s,s_v,2), (s_v,v,3), (v,z,4), \\
    \label{eq:p2}P_2 = (s,v,1), (v,z_v,2), (z_v,z,3).
\end{gather}
Now we construct a set $C \in V(H)$ as follows. For each $v \in V(H)$:
\begin{itemize}
    \item If more than one of the three vertices $s_v$, $v$, and $z_v$ belong to $S$, then add $v$ to $C$.
    \item If only one of the three vertices $s_v$, $v$, and $z_v$  belongs to $S$, then do not add $v$ to $C$.
\end{itemize}

First, based on the fact that at least one of the three vertices $s_v$, $v$, and $z_v$ belongs to $S$, it is clear that $|C| \leq k$. Second, if there is an edge $(u,v) \in E(H)$ such that $\{u,v\} \subseteq V(H) \backslash C$, following the previous claims, it results in both path $P_3$ and $P_4$ (which are shown in the equations \ref{eq:p3} and \ref{eq:p4} respectively) being present in a temporal subgraph induced by $V\backslash S$. Therefore $C$ is a vertex cover with cardinality at most $k$.
\begin{gather}
    \label{eq:p3}P_3 = (s,s_v,2), (s_v,z_u,3), (z_u,z,4), \\
    \label{eq:p4}P_4 = (s,s_u,2), (s_u,z_v,3), (z_v,z,4).
\end{gather}
\end{proof}

Since every temporal path from $s$ to $z$ contains more than two edges, then $\emptyset$ is a strict $(s,z,1)$-temporal separator. Since every strict $(s,z,2)$-temporal path is of the form $(s,v,t), (v,z,t+1)$, 
the Strict $(s,z,2)$-Temporal Separator problem could be solved in polynomial time easily. The Strict $(s,z,t)$-Temporal Separator problem on a graph $G = (V,E,\tau)$ with $\tau = t$ is the same as the Strict $(s,z)$-Temporal Separator. Therefore, in case $\tau  = t = 3$ this problem is equivalent to the Strict $(s,z)$-Temporal Separator problem with $\tau = 3$. Zschoche et al. \cite{zschoche2020complexity} present a polynomial time algorithm for finding a minimum strict $(s,z)$-temporal separator on a temporal graph $G = (V,E,\tau)$ when $\tau < 5$. So, this case could be solved in polynomial time. Although we know that finding a strict $(s,z,t)$-temporal separator on a temporal graph $G = (V,E,3)$ is polynomial-time solvable with the algorithm which is presented in \cite{zschoche2020complexity}, we describe another simple algorithm to solve this problem. 

In the first step of the algorithm, we check if there is an edge between $s$ and $t$. If so, it is clear that there are no separator sets because the direct path using this edge from $s$ to $z$ will remain with the removal of any node from the graph.

Next, for every temporal path from $s$ to $z$ of length two, such as $(s,x, t_1), (x,z$ $,t_2)$ with $t_2 = t_1+1$, it is clear that we have to remove $x$ if we want to remove this path from the graph. So, it is clear that $x \in S$.

In the last step, we know that the length of every temporal path in the graph is three. So, every path from $s$ to $z$ should be of the following form:
    \[(s,x, 1), (x,y, 2), (y,z, 3).\]
Now, put every node $x$ with existing edge $(s,x)$ into the set $X$ with time label 1. Also, put every node $y$ that is a neighbor of $z$ into the set $Y$ with time label 3. Now, it is clear that $X \cap Y = \emptyset$, for otherwise there exists a node $u$ with two existing edges $e_1 = (s,u, 1)$ and $e_2 = (u,z, 3)$, while this node should be removed in the previous step. Therefore, every strict temporal path from $s$ to $z$ should have a corresponding edge $(x,y, 2)$ where $x\in X$ and $y\in Y$. So, we should remove either $x$ or $y$ for every edge $(x,y, 2)$, where $x \in X$ and $y \in Y$. In order to do this we could use any known polynomial time algorithm for the Vertex Cover problem in bipartite graphs.

In the rest of this section we show $\Omega(\log n + \log(\tau))$-inapproximability (assuming $\mathcal{NP}\subset\mbox{\sc Dtime}(n^{\log\log n})$) for the $(s,z,t)$-Temporal Separator problem. This is proved by a strict reduction from the Set Cover problem. Recall that in the Set Cover problem, one is given a collection $\mathcal{S}$ of subsets of a universe $U$ that jointly cover the universe. The goal is to find a minimum size sub-collection of $\mathcal{S}$ that covers $U$.

\begin{theorem}
\label{thm:set_cover_red}
    For every $t > 0$ there is a strict polynomial time reduction from the Set Cover problem to the $(s,z,t)-$Temporal Separator problem.
\end{theorem}

\begin{proof}
    Let $(U, \mathcal{S})$ be an instance of the Set Cover problem, where $U = \{1,2,\dots n\}$ is the universe and $\mathcal{S} = \{S_1,S_2,\dots, S_m\}$ is a family of sets the union of which covers $U$. For each $i \in U$ define the family $\mathcal{F}_i$ as $\mathcal{F}_i = \{ S \in \mathcal{S} \mid i \in S\},$
    i.e., $\mathcal{F}_i$ consists of all sets from $\mathcal{S}$ that contain element $i$. Let $k_i = |\mathcal{F}_i|$ and order the elements of each $\mathcal{F}_i$ in the order of increasing indices, i.e.,
    \begin{equation}\label{eq:fi}
    \mathcal{F}_i = \{ S_{i_1}, \ldots, S_{i_{k_i}}\}.
    \end{equation}
    
    Our reduction outputs a temporal graph $f(U,\mathcal{S}) = (V\cup\{s,z\}, E)$ where:
    \begin{itemize}
        \item the vertex set is $V\cup \{s,z\}=\{v_i | i \in [m]\}\cup \{s,z\}$;
        \item the edge set is $E=E_1 \cup E_2 \cup \cdots \cup E_n$, where
        \[E_i = \{(s, v_{i_1}, i \cdot t), (v_{i_1}, v_{i_2}, i \cdot t), \dots, (v_{i_{k_i-1}}, v_{i_{k_i}}, i \cdot t), (v_{i_{k_i}}, z, i \cdot t)\}.\]
    \end{itemize}

    The main idea behind the proof is to map every element of $U$ to a path from $s$ to $z$ in $f(U,\mathcal{S})$ bijectively, so by covering an element, we remove the corresponding path in $f(U,\mathcal{S})$ as well as by removing a path we cover the corresponding element.
    
    We claim that $V'=\{v_{j_1}, \ldots, v_{j_\ell} \} \subseteq V$ is a $(s,z,t)-$temporal separator for $f(U, \mathcal{S})$ if and only if $\mathcal{S}'=\{ S_{j_1}, \ldots, S_{j_\ell}\} \subseteq \mathcal{S}$ is a set cover for $(U, \mathcal{S})$.
    
    \begin{figure}
        \centering
        \includegraphics[scale=0.7]{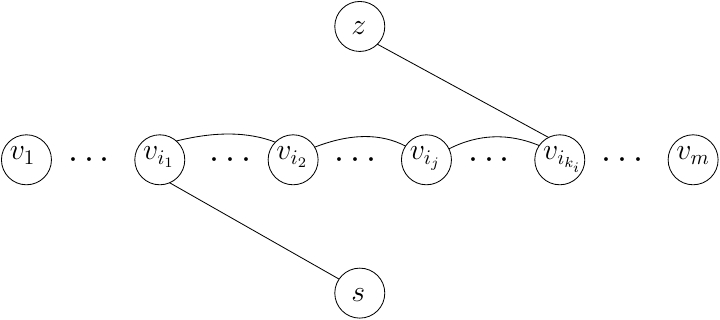}
        \caption{Layer $G_{i\cdot t}$ of the temporal graph  used in the proof of Theorem~\ref{thm:set_cover_red}.}
        \label{fig:E_I}
    \end{figure}
    
    Figure \ref{fig:E_I} represents the edges in the layer $G_{it}$, which contain all the edges in $E_i$. It illustrates that element $i$ in the universe $U$ corresponds to a path $E_{i}$, as well as the element $i$ is covered by the set $S_{i_j} \in \mathcal{S'}$ if and only if a temporal path which is shown in Figure \ref{fig:E_I} is removed from the temporal graph by removing the vertex $v_{i_j} \in V'$.
    
    $\rightarrow$ Suppose for contradiction that $\mathcal{S}'$ does not cover $U$. Pick an arbitrary item $i \in U$ that is not covered and consider the following path 
        $P = [ (s, v_{i_1}, i\cdot t), (v_{i_1}, v_{i_2}, i\cdot t), \dots, (v_{i_{k_i-1}}, v_{i_{k_i}}, i\cdot t), (v_{i_{k_i}}, z, i\cdot t)],$
    where the indices are according to \eqref{eq:fi}. Since $i$ is not covered, $\mathcal{F}_i \cap \mathcal{S}' = \emptyset$, so $P$ is present in $f(U, \mathcal{S}) \setminus V'$ violating the assumption that $V'$ is a $(s,z,t)-$temporal separator (note that $\ttime(P) = 0$).
    
    $\leftarrow$ Now, suppose for contradiction that $V'$ is not a $(s,z,t)$-temporal separator. Thus, there is path $P$ from $s$ to $z$ with $\ttime(P) < t$. From the definition of $f(U, \mathcal{S})$ it is clear that $P$ should be using edges only from $E_j$ for some $j \in [n]$. Note that there is a unique $(s,z)$-temporal path that can be constructed from $E_j$, namely, 
            $P = [(s, v_{j_1}, j \cdot t), (v_{j_1}, v_{j_2}, j \cdot t), \dots, (v_{j_{k_j-1}}, v_{j_{k_j}}, j \cdot t), (v_{j_{k_j}}, z, j \cdot t)].$
        This implies that element $j$ is not covered by $\mathcal{S}'$, since otherwise, one of the $v_{j_i}$ would be in $V'$.
    
    Following the previous claim, every solution in $(s,z,t)$-Temporal Separator has a corresponding solution in Set Cover, and vice versa. Therefore, an optimal solution in $(s,z,t)$-Temporal Separator, has a corresponding optimal solution in Set Cover. As a result $\frac{|V'|}{|V_{opt}|} = \frac{|S'|}{|S_{opt}|}$.
    This implies that the reduction is strict.
\end{proof}

Due to the inapproximability of Set Cover (see \cite{feige1998threshold}), we have the following:

\begin{corollary}
The $(s,z,t)$-Temporal Separator problem is not approximable to within $(1 - \epsilon)(\log n + \log(\tau))$ in polynomial time for any  $\varepsilon >0$, unless $\mathcal{NP}\subset\mbox{\sc Dtime}(n^{\log\log n})$.
\end{corollary}


\subsection{Approximation Algorithms}
\label{section:approximation}
In this section, we present an efficient $\tau^2$-approximation for the $(s,z,t)$-Temporal Separator problem. We begin by establishing a $\tau$-approximation for the $(s,z)$-Temporal Separator problem. The main tool used in this section is the \emph{flattening}\footnote{The concept of flattening is not new, and it is similar to the static expansion of a temporal graph -- see, for example, \cite{mertzios}.} of a temporal graph $G=(V,E,\tau)$ with respect to vertices $s$ and $z$, denoted by $F(G, s, z) = (V',E')$. To ease the notation we omit the specification of $s$ and $z$ and denote the flattening of $G$ by $F(G)$. The flattening $F(G)$ is a static directed graph defined as follows: the vertex set $V'$ is the union of $\tau$ disjoint sets $V_1, V_2, \dots, V_{\tau}$ and $\{s,z\}$, where each $V_i$ is a disjoint copy of $V-\{s,z\}$. Denoting the vertices of $V$ by $v_1, v_2, \ldots, v_n$, we have $\forall i \in [\tau]\;\; V_i = \{v_{j,i} | v_j \in V - \{s,z\}\}$. The edge set $E'$ of the flattening $F(G)$ is defined as follows:
\begin{itemize}
    \item For each $(v_i, v_j, t') \in E$ with $v_i, v_j \not\in \{s,z\}$ we add edges $(v_{i,t'}, v_{j, t'})$ and $( v_{j, t'}, v_{i,t'})$ to $E'$.
    \item For each $v_i \in V$ and each time $t' \in [\tau - 1]$ we add an edge $(v_{i,t'}, v_{i, t'+1})$ to $E'$.
    \item For each $(s, v_i, t') \in E$ we add an edge $(s, v_{i,t'})$ to $E'$.
    \item For each $(z, v_i, t')$ we add an edge $(v_{i,t'}, z)$ to  $E'$.
\end{itemize}
Clearly, $F(G)$ is defined to express temporal $(s,z)$-paths in $G$ in terms of $(s,z)$-paths in $F(G)$. More specifically, if we have a temporal $(s,z)$ path $P$ in $G$ then there is an analogous static $(s,z)$ path $P'$ in $F(G)$. If $P$ begins with an edge $(s, v_i, t_1)$ then $P'$ begins with an edge $(s, v_{i,t_1})$. After that if the next edge in $P$ is  $(v_i, v_j, t_2)$, we can simulate it in $F(G)$ by introducing a sequence of edges $(v_{i, t_1}, v_{i, t_1+1}), (v_{i, t_1+1}, v_{i, t_1+2}), \ldots, (v_{i, t_2-1}, v_{i, t_2})$ followed by an edge $(v_{i, t_2}, v_{j, t_2})$, and so on until the vertex $z$ is reached. This correspondence works in reverse as well. If $P'$ is a static $(s,z)$ path in $F(G)$ then we can find an equivalent temporal $(s,z)$ path in $G$ as follows. If the first edge in $P'$ is $(s, v_{i, t_1})$ then this corresponds to the first edge of $P$ being $(s, v_i, t_1)$. For the following edges of $P'$, if the edge is of the form $(v_{i, t'}, v_{i, t'+1})$ then it is simply ignored for the purpose of constructing $P$ (since it corresponds to the scenario where the agent travelling along the path is simply waiting an extra time unit at node $v_i$), and if the edge is of the form $(v_{i, t'}, v_{j, t'})$ then we add the edge $(v_i, v_j, t')$ to $P$. This continues until $z$ is reached. Thus, there is a temporal $(s,z)$ path $P$ in $G$ if and only if there is a static $(s,z)$ path $P'$ in $F(G)$. Moreover, if $S$ represents the internal nodes of the path $P$ then we can find $P'$ with internal nodes in $\bigcup_{t' \in [\tau]} \{v_{i, t'} : v_i \in S\}$. In the reverse direction, if $P'$ uses internal nodes $S'$ then we can find $P$ with internal nodes in $\{v_i : \exists t' \;\; v_{i,t'} \in S'\}$.

Armed with these observations, we show that the sizes of $(s,z)$-temporal separators in $G$ and $(s,z)$-separators (non-temporal) in $F(G)$ are related as follows.

\begin{theorem}
\label{theorem:appx1}
\begin{enumerate}
\item If $S$ is an $(s,z)$-temporal separator in $G$ then there is an $(s,z)$-separator of size at most $\tau|S|$ in $F(G)$.
\item If $S'$ is an $(s,z)$-separator in $F(G)$ then there is an $(s,z)$-temporal separator of size at most $|S'|$ in $G$.
\end{enumerate}    
\end{theorem}
\begin{proof}

\begin{enumerate}
\item Define $S' = \bigcup_{t' \in [\tau]} \{v_{i,t'} \in V' : v_i \in S\}$. Clearly, $|S'| = \tau |S|$. Suppose for the contradiction that $S'$ is not an $(s,z)$-separator in $F(G)$. Then there is a path from $s$ to $z$ in $F(G)$ that avoids all vertices in $S'$. By the observations made prior to the statement of this theorem, this path corresponds to a temporal path in $G$ that avoids vertices in $S$. Thus, $S$ is not an $(s,z)$-temporal separator in $G$.
\item Define $S = \{v_i \in V : \exists t' \in [\tau] \;\; v_{i,t'} \in S'\}$. Clearly, $|S| \le |S'|$. An argument similar to the one given in the previous part establishes that $S$ is an $(s,z)$-temporal separator in $G$.
\end{enumerate}
\end{proof}

\begin{corollary}
The $(s,z)$-Temporal Separator problem on a temporal graph $G=(V,E, \tau)$ can be approximated within $\tau$ in $O((m+n\tau) n\tau)$ time, where $n = |V|$ and $m= |E|$.
\end{corollary}
\begin{proof}
We can use any existing efficient algorithm to solve the $(s,z)$ separator problem on $F(G)$ and return its answer, which will give $\tau$-approximation by Theorem~\ref{theorem:appx1}. For example, the stated runtime is achieved by applying Menger's theorem and the Ford-Fulkerson algorithm to compute the maximum number of vertex-disjoint paths in $F(G)$. Then the running time is $O(|E'||V'|)$. Observing that $|E'| \le |E| + |V| \tau$ and $|V'| \le  |V| \tau$, finishes the proof of this corollary.
\end{proof}

Next, we describe how the $(s,z,t)$-Temporal Separator problem can be approximated using a slight extension of the above ideas. First, for a temporal graph $G = (V, E, \tau)$ and two integers $t_1 \le t_2$ we define $E[t_1:t_2] = \{ (u, v, t) \in E : t_1 \le t' \le t_2\}$. We also define $G[t_1 : t_2] = (V, E[t_1 : t_2], t_2)$, which can be thought of as graph $G$ restricted to time interval $[t_1, t_2]$. The idea behind approximating a minimum $(s,z,t)$-temporal separator is to combine $(s,z)$-temporal separators of $F(G[1:t+1]), F(G[2:t+2]), \ldots, F(G[\tau-t : \tau])$.

\begin{theorem}
The $(s,z,t)$-Temporal Separator problem on a temporal graph $G=(V,E, \tau)$ can be approximated within $\tau^2$ in $O((m+n\tau)n\tau^2)$ time, where $n = |V|$ and $m = |E|$.
\end{theorem}
\begin{proof}
The algorithm has essentially been described prior to the statement of the theorem, so the running time is clear. It is left to argue that it produces $\tau^2$-approximation. This can be argued similarly to Theorem~\ref{theorem:appx1}.
\begin{enumerate}
\item Let $S$ be a $(s,z,t)$-temporal separator in $G$. Then for $G[i: i+t]$ we define $S_i$ to consist of all nodes $v_{j, t'}$ with $v_j \in S$. Since $S$ removes all paths from $G$ of travelling time $\le t$ and $G[i:i+t]$ only has paths of travelling time $\le t$, then $S_i$ is a $(s,z)$-separator in $G[i:i+t]$ of size $|S_i| = \tau |S|$. Thus, if there is an $(s,z,t)$-temporal separator of size $|S|$ in $G$ then the combined size of all $(s,z,t)$-temporal separators of $G[1:t+1], G[1:t+2], \ldots, G[\tau-t, \tau]$ is at most $\tau^2 |S|$.
\item Let $S_i$ be a $(s,z)$-temporal separator in $G[i:i+t]$. Define $S = \{ v_j : \exists i \exists t' \;\; v_{j, t'} \in S_i\}$. It is easy to see that $S$ is a $(s,z,t)$ temporal separator in $G$. Paths of travelling time at most $t$ that begin with an edge $(s, v_i, t_1)$ are present in $G[t_1, t_1 + t]$, and so removal of $S_{t_1}$ removes such temporal paths in $G[t_1, t_1+t]$. Since $S_{t_1}$ is ``projected'' onto $V$ and included in $S$, these paths are eliminated from $G$.
\end{enumerate}
\end{proof}


\section{Temporal Separators with Deadlines on Special Families of Graphs}
\label{sec:special-families}

\subsection{Temporal Graphs with Branchwidth at most $2$}
\label{sec:bounded_branchwidth}

The graphs with branchwidth $2$ are graphs in which each biconnected component is a series-parallel graph \cite{robertson1991graph}. In this section, we present an efficient algorithm to solve the $(s,z,t)$-Temporal Separator problem on temporal graphs whose underlying static graphs have branchwidth at most $2$. In fact, our algorithm works for a more general class of problems, which we refer to as ``restricted path $(s,z)$-Temporal Separator.'' The goal in this more general problem is to select a set of vertices $S$ such that the removal of $S$ from the given temporal graph $G$ removes all $(s,z)$ paths in a restricted family of paths. The $(s,z,t)$-Temporal Separator problem is seen as a special case of this, where paths are restricted to have travelling time less than $t$. Restricted family of paths could be any path family implicitly defined by a procedure $ExistsRestrictedPath(G, s, z)$ which takes as input a temporal graph $G$, a pair of nodes $s$ and $z$, and returns true if and only if there exists a restricted temporal path between $s$ and $z$ in $G$. Due to Lemma~\ref{lemma:temporalPath}, we know that such a procedure exists in the case of temporal paths restricted by travelling time, which is suitable for the $(s,z,t)$-Temporal Separator problem.

For the rest of this section, we assume that $G$ is a temporal graph such that $bw(G_\downarrow) \le 2$ unless stated otherwise. Furthermore, we assume that $G_\downarrow$ is connected, otherwise, if $s$ and $z$ belong to different connected components the answer to the problem is trivially $\emptyset$, and if they belong to the same connected component, the problem reduces to analyzing that connected component alone. We introduce some notation and make several observations about branch decomposition before we give full details of our algorithm. Recall from Section~\ref{sec:prelim} that branch decomposition of $G$ of width $2$ can be computed in linear time. Thus, we assume that the algorithm has access to such a decomposition, which we denote by $(T, \beta)$. We use $\rho$ to denote the root of $T$ and we define the function $top: V(G) \rightarrow V(T)$ as follows.  For $v \in V(G)$  we let $top(v)$ be the furthest node  $x \in V(T)$ from the root $r$ which satisfies $E(v) \subseteq \beta(x)$. We also use $x_l$ to denote the left child of $x$ and $x_r$ to denote the right child of $x$. For a node $x \in V(T)$ we define $G^{in}_x$ to be the temporal graph obtained from $G$ by keeping only those edges $(u, v, t)$ with $(u,v) \in \beta(x)$ and removing all vertices of degree $0$. We collect several useful observations about the introduced notions in the following lemma.

\begin{lemma}
\label{lem:helper_bw_2}
\begin{enumerate}
\item If $v \in \partial \beta(x)$ then $v \in \partial \beta(x_\ell)$ or $v \in \partial \beta(x_r)$.
\item If $top(v) = x$ then $v \in \partial \beta(x_\ell)$ and $v \in \partial \beta(x_r)$.
\item If $v \in V(G^{in}_x) \setminus \partial \beta(x)$ then all edges incident on $v$ in $G$ are present in $G^{in}_x$.
\end{enumerate}
\end{lemma}
\begin{proof}
\begin{enumerate}
\item Since $v \in \partial \beta(x)$ it means that some but not all edges incident on $v$ in $G$ appear in $\beta(x)$. Since $\beta(x) = \beta(x_\ell) \cup \beta(x_r)$, it implies that some but not all edges incident on $v$ must appear either in $\beta(x_\ell)$, or $\beta(x_r)$, or both.
\item If $top(v) = x$ then $E(v) \subseteq \beta(x)$. Suppose for contradiction that $v \not\in \partial \beta(x_\ell)$. This can happen for two reasons: either (1) $E(v) \subseteq \beta(x_\ell)$, or (2) $E(v) \cap \beta(x_\ell) = \emptyset$. In case (1) we obtain a contradiction with the definition of $top(v)$ since  $x_\ell$ is further from the root than $x$ and it still contains all of $E(v)$. In case (2) observe that we must have $E(v) \subseteq \beta(x_r)$, thus obtaining a contradiction with the definition of $top(v)$ again since $x_r$ is further from the root than $x$ and it still contains all of $E(v)$.
\item Since $v \in V(G^{in}_x) \setminus \partial \beta(x)$ it means that there is at least one edge incident on $v$ in $V(G^{in}_x)$. Since $v$ is not in the boundary of $\beta(x)$, it means that all edges incident on $v$ in $G$ must be present in $\beta(x)$.
\end{enumerate}
\end{proof}

\begin{algorithm}[H]
\label{alg:sep_bw_2}
 \caption{This algorithm finds a restricted $(s,z)-$temporal separator in a temporal graph $G$ with $bw(G_\downarrow) \le 2$.}
\SetAlgoLined
\SetKwFunction{RTS}{RTS}
\SetKwFunction{ExistsRestrictedPath}{ExistsRestrictedPath}
\SetKwProg{Fn}{Function}{:}{}
\Fn{\RTS{$G, s, z$}}{
    \If{\ExistsRestrictedPath{$G,s,z$}=false}{
        \Return $\emptyset$\;
    }
    \For{$v \in V(G) \setminus \{s,z\}$}{
        \If{\ExistsRestrictedPath{$G\setminus \{v\},s,z$}=false}{
            \Return $\{v\}$\;
        }
    }
    \uIf{$top(s)=top(z)$}{
        \Return \RTS{$G^{in}_{\rho_\ell}, s, z$} $\cup$ \RTS{$G^{in}_{\rho_r}, s, z$}\;
    }
    \uElseIf{$top(s), top(z)$ are not ancestors of each other}{
        \Return $\partial \beta(top(z))$\;
    }
    \Else{ 
        \tcc{assume $top(z)$ is ancestor of $top(s)$, otherwise swap $s$ and $z$}
        \uIf{$z \not\in \partial \beta(top(s))$}{
            \Return $\partial \beta(top(s))$\;
        }
        \uElseIf{$\partial \beta(top(s))=\{z\}$}{
            \Return \RTS{$G^{in}_{top(s)}, s, z$}\;
        }
        \Else{
            \tcc{$\partial \beta(top(s)) = \{z,q\}, \partial \beta(top(s)_\ell) = \{s, z\}, \partial \beta(top(s)_r) = \{s, q\}$}
            $S \gets$ \RTS{$G^{in}_{top(s)_\ell}, s, z$}\;
            \eIf{\ExistsRestrictedPath{$G\setminus S,s,z$}}{
                \Return $S \cup \{q\}$\;
            } {
                \Return $S$\;
            }
        }
    }    
}
 \end{algorithm}

Now, we are ready to describe our algorithm, which is denoted by $RTS$. The algorithm starts by checking if there is a restricted temporal path from $s$ to $z$ in $G$, and if such a path does not exist then the algorithm immediately returns $\emptyset$. Then the algorithm checks if there exists a restricted temporal separator of size $1$ by testing whether there is a restricted temporal path in $G \setminus \{v\}$ for each $v \in V(G) \setminus\{s,z\}$. Then the algorithm computes $top(s)$ and $top(z)$ and the computation splits into three cases: (1) if $top(s) = top(z)$; (2) if $top(s)$ and $top(z)$ are not on the same root-to-leaf path in $T$ (i.e., neither one is an ancestor of another); and (3) if one of $top(s), top(z)$ is an ancestor of another. We shall later see that case (1) implies that $top(s) = top(z) = \rho$. In this case, the algorithm invokes itself recursively on the two subtrees of $T$ -- the subtree rooted at the left child of $\rho$ and the subtree rooted at the right child of $\rho$. The separators obtained on these two subtrees correspond to separators of $G^{in}_{\rho_\ell}$ and $G^{in}_{\rho_r}$ and their union is returned as the separator for $G$. In case (2) the algorithm returns the boundary of $\beta(top(z))$ (it could return the boundary of $\beta(top(s))$ instead -- it does not make a difference) as the answer. In case (3), we assume without loss of generality that $top(z)$ is the ancestor of $top(s)$, and handling of this case depends on whether $z$ belongs to the boundary of $\beta(top(s))$ or not. In fact, this case splits into three subcases: (3.1) if $z \not\in \partial \beta(top(s))$ then the algorithm immediately returns $\partial \beta(top(s))$; (3.2) if $\partial \beta(top(s)) = \{z\}$ then the algorithm invokes itself recursively on $G^{in}_{top(s)}$; and (3.3) if $\partial \beta(top(s)) = \{z,q\}$ for some vertex $q \neq s, z$ then the algorithm first invokes itself recursively on $G^{in}_{top(s)_\ell}$ (assuming $\partial \beta(top(s)_\ell) = \{s,z\}$) and stores the answer in $S$. If $S$ proves to be a separator in $G$ then $S$ is returned, otherwise, $q$ is added to $S$ and returned. The pseudocode is presented in Algorithm~\ref{alg:sep_bw_2}.

 \begin{theorem}
 Algorithm~\ref{alg:sep_bw_2} correctly computes a minimum-sized restricted path $(s,z)$-temporal separator for a temporal graph $G$ such that $bw(G_\downarrow) \le 2$.
 \end{theorem}
 \begin{proof}
 The proof proceeds by the case analysis reflecting the structure of the algorithm. Clearly, the algorithm correctly identifies when there is a separator of size $0$ or $1$ since it performs brute-force checks for these special cases. Assuming that there is no separator of size $\le 1$, we discuss the correctness for the remaining three cases.

 Case (1): $top(s) = top(z) = x \in V(T)$. Observe that Lemma~\ref{lem:helper_bw_2}, item 1, implies that $s, z \in \partial \beta(x_\ell)$ and $s, z \in \partial \beta(x_r)$. Since the branchwidth is $2$, it implies that $\partial \beta(x_\ell) = \partial \beta(x_r) = \{s, z\}$. In addition, we know that $s, z \not\in \partial \beta(x)$ by the definition of $top()$. And since every vertex in $\partial \beta(x)$ must appear in $\partial \beta(x_\ell)$ or $\partial \beta(x_r)$ (using Lemma~\ref{lem:helper_bw_2}, item 2), we conclude that $\partial \beta(x) = \emptyset$. By Lemma~\ref{lem:helper_bw_2}, item 3,  every vertex in $G^{in}_x$ has all its edges from $G$. Therefore $G^{in}_x$ is disconnected from the rest of $G$. However, we assume that $G$ is connected, so we must have $G^{in}_x = G$. This is true only when $x = \rho$. Thus, we must have in this case that $top(s) = top(z) = \rho$. Observe that if $P$ is a restricted temporal path between $s$ and $z$ (that does not have $s$ or $z$ as intermediate nodes) then it cannot use edges from both $\beta(\rho_\ell)$ and $\beta(\rho_r)$. Suppose, for contradiction, that $P$ uses both kinds of edges, then there must be a vertex $v$ on this path incident on $e_1$ and $e_2$ such that $e_1 \in \beta(x_\ell)$ and $e_2 \in \beta(x_r)$. Since $\beta(x_\ell), \beta(x_r)$ partition all the edges, it implies that $e_2 \not\in \beta(x_\ell)$. This means that $v \in \partial \beta(x_\ell) = \{s,z\}$, but $v \neq s, z$, giving a contradiction. Therefore, the minimum size restricted path temporal separator in $G$ is the union of minimum size restricted path temporal separators in $G^{in}_{\rho_\ell}$ and $G^{in}_{\rho_r}$, which is precisely what our algorithm outputs.

 Case (2): $top(s)$ and $top(z)$ do not lie on the same root-to-leaf path in $T$. One of the consequences of Lemma~\ref{lem:helper_bw_2}, item 3, is that removing $\partial \beta(x)$ from $G$ separates all vertices in $V(G^{in}_x)$ from the rest of the graph. Therefore, removing $\partial \beta(top(z))$ separates all vertices in $G^{in}_{top(z)}$ from the rest of the graph. Observe that $z \in V(G^{in}_{top(z)})$ and $s \not\in V(G^{in}_{top(z)})$ (by the condition of this case). Therefore removing $\partial \beta(top(z))$ separates $s$ from $z$. We claim that this is the minimum separator in this case. This is because when this line is reached we are guaranteed that there is no separator of size $1$, and $|\partial \beta(top(z))| \le 2$ (in fact, it must be then equal to $2$). We only need to be careful that neither $z$ nor $s$ is in $\partial \beta(top(z))$, but it is clear from the definition of $top()$ and the case condition.

 Case (3): $top(z)$ is an ancestor of $top(s)$ (if $top(s)$ is an ancestor of $top(z)$ then we can exchange the roles of $s$ and $z$ for the sake of the argument). This case has three subcases.

 Subcase (3.1): $z \not\in \partial \beta(top(s))$. This is similar to case (2) described above. The algorithm can return $\partial \beta(top(s))$ as a minimum size separator.

 Subcase (3.2): $\partial \beta(top(s))=\{z\}$. In this case, the structure of the graph is such that $G^{in}_{top(s)}$ is connected to the rest of the vertices in $G$ via the node $z$, while vertex $s$ lies in $G^{in}_{top(s)}$. Thus, to separate $z$ from $s$, it is sufficient to separate them in $G^{in}_{top(s)}$, which is what the algorithm does.

 Subcase (3.3); $\partial \beta(top(s)) = \{z, q\}$. By Lemma~\ref{lem:helper_bw_2}, item 2, it follows that $s \in \partial \beta(top(s)_\ell)$ and $s \in \partial \beta(top(s)_r)$. By Lemma~\ref{lem:helper_bw_2}, item 1, it follows that $z, q \in \beta(top(s)_\ell) \cup \beta(top(s)_r)$. Since branchwidth is at most $2$, we have (without loss of generality) that $\partial \beta(top(s)_\ell) = \{s, z\}$ and $\partial \beta(top(s)_r) = \{s, q\}$. By an argument similar to the one in case (1), we can establish that any restricted $(s,z)$ temporal path (that does not use $s$ or $z$ as intermediate nodes) must either consist entirely of edges in $\beta(top(s)_\ell)$ or entirely of edges in $\beta(top(s)_r)$. Thus, we can compute the two separators and take their union; however, we can simplify the calculation observing that the only separator we need to consider for the $G^{in}_{top(s)_r}$ is $\{q\}$, since $G^{in}_{top(s)_r}$ is connected to the rest of $G$ only through $q$ and $s$.
 \end{proof}


\begin{corollary}
Given a temporal graph $G = (V, E, \tau)$ with $bw(G_\downarrow) \le 2$, the problem $(s, z, t)$-Temporal Separator is solvable in time $O(|V ||E||\mathcal{T}|)$ where $\mathcal{T} = \{t(e) : e \in E(s)\}$.
\end{corollary}

\subsection{Temporal Graphs with a ``Tree-like'' Underlying Graph}
\label{sec:tree_like}

In this section, we present a polynomial time greedy algorithm (motivated by the point-cover interval problem) for computing a path restricted $(s,z)$-temporal separator (see Section~\ref{sec:bounded_branchwidth}) on a temporal graph $G$ such that $G_{\downarrow}\setminus\{s,z\}$ is a tree if the existence of a restricted $(s,z)$-temporal path could be checked in polynomial time. 

We assume that we are given a temporal graph $G$ such that $G_\downarrow \setminus\{s,z\}$ is a tree, which we denote by $T$. For a pair of nodes $(u,w)$, we let $P_{u,w}$ denote the unique shortest path in $T$ between $u$ and $w$. For a vertex $v \in V(T)$, we define a removal list of $v$, denoted by $RL_v$, to consist of all unordered pairs $(u,w)$ such that $v \in V(P_{u,w})$ and there exists a restricted $(s,z)$-temporal path in $G$ using the edges of $P_{u,w}$. For a pair $u,w \in V(T)$, we define two temporal graphs: (1) $G^1_{u,w}$ is $G$ induced on the edges of $E(P_{u,w}) \cup \{(s,u), (v,z)\}$, and (2) $G^2_{u,w}$ is $G$ induced on the edges of $E(P_{u,w}) \cup \{(s,v),(u,z)\}$. The removal lists for all vertices in $V(T)$ can be computed efficiently as follows. Initialize all removal lists to be empty. For each pair of vertices $u,w \in V(T)$ check if there is any restricted $(s,z)$-temporal path in $G^1_{u,w}$ or $G^2_{u,w}$, and if so, then add $(u,w)$ to the removal lists of all nodes in $P_{u,w}$. Let $\mathcal{U} = \bigcup_{v \in V(T)} RL_v$ be the set of all pairs of nodes that appear in removal lists. The following observation is immediate from the definitions and shows that computing a minimum size restricted path $(s,z)$-temporal separator reduces to covering $\mathcal{U}$ with as few removal lists as possible.

\begin{observation}
\label{obs:tree-removing-list}
A set of $S$ is a restricted path $(s,z)$-temporal separator if and only if $\bigcup_{v \in S} RL_v = \mathcal{U}$.
\end{observation}

A vertex $v$ is called topmost if there exists a pair $(u,w) \in RL_v$ such that $(u,w) \not\in RL_{parent(v)}$. Our greedy algorithm, called $GreedyRTS$, starts out with an empty solution $S = \emptyset$, and then adds more vertices to $S$ as follows. While there are non-empty removal lists, the algorithm selects a topmost vertex $v$ with maximum distance from the root of $T$, adds $v$ to the set $S$, and removes all pairs in $RL_v$ from the removing lists of all the other vertices. The pseudocode is given in Algorithm~\ref{alg:tree}.

\begin{algorithm}[]
 \caption{This algorithm computes a minimum sized restricted path $(s,z)$-temporal separator in a temporal graph $G$ when $G_\downarrow \setminus \{s,z\}$ is a tree $T$.}\label{alg:tree}
\SetAlgoLined
\SetKwFunction{ComputeRLs}{ComputeRLs}
\SetKwFunction{ExistsRestrictedPath}{ExistsRestrictedPath}
\SetKwFunction{GreedyRTS}{GreedyRTS}
\SetKwProg{Fn}{Function}{:}{}
\Fn{\ComputeRLs{$G, s, z$}}{
    $\mathcal{U} \gets \emptyset$\;
    \For{ $(u,w) \in V(T) \times V(T)$}{
        \If{\ExistsRestrictedPath{$G^1_{u,w}, s, z$} or \ExistsRestrictedPath{$G^2_{u,w}, s, z$}}{
            $\mathcal{U} \gets \mathcal{U} \cup \{(u,w)\}$\;
            \For{$v \in V(P_{u,w})$}{
                $RL_v \gets RL_v \cup \{(u,w)\}$\;                
            }
        }
    }
}
\Fn{\GreedyRTS{$G, s, z, RL, \mathcal{U}$)}}{
    $S \gets \emptyset$\;
    \While{$\mathcal{U} \neq \emptyset$}{
        $v \gets$ furthest node from the root of $T$ such that $\exists (u,w) \in RL_v \setminus RL_{parent(v)}$\;
        $S \gets S \cup \{v\}$\;
        $\mathcal{U} \gets \mathcal{U} \setminus RL_v$\;
        \For{$w \in V(T)$}{
            $RL_w \gets RL_w \setminus RL_v$\;
        }
    }    
    \Return $S$\;
}
\end{algorithm}

\begin{theorem}
\label{theorem:solvabality_in_tree_base}
Algorithm~\ref{alg:tree} computes a minimum-sized restricted path $(s,z)$-temporal separator in a temporal graph $G$ with $G_\downarrow \setminus \{s,z\}$ being a tree.
\end{theorem}

\begin{proof}
Let $S$ denote the solution produced by $GreedyRTS$. It is clear from the algorithm's description and Observation~\ref{obs:tree-removing-list} that $S$ is, indeed, a restricted path $(s,z)$-temporal separator. It is only left to show that there is no smaller temporal separator. Consider the order of vertices in $S$ in which they are included by $GreedyRTS$. For a minimum-sized separator $S_{opt}$ we can define $k$ to be the largest integer such that $S$ and $S_{opt}$ agree on the first $k$ vertices considered by $GreedyRTS$. Now, we fix a particular minimum size separator $S_{opt}$ that maximizes $k$. We will show that $k = |S|$, establishing the claim. Suppose for contradiction that $k < |S|$. Therefore, $S$ and $S_{opt}$ disagree on the $(k+1)^\text{st}$ vertex $x$, i.e., $x \in S$ and $x \not\in S_{opt}$. 

Vertex $x$ was selected by $GreedyRTS$ since there is a pair of nodes $u, w$ such that $(u,w) \in RL_x \setminus RL_{parent(x)}$. Since $(u,w)$ has not been removed from $\mathcal{U}$ at the time when $x$ was considered and $S_{opt}$ agreed with $S$ up until that point, it means that there must be some other vertex $x' \in S_{opt}$ such that $(u,w) \in RL_{x'}$. We claim that $RL_{x'} \subseteq RL_x$. First, observe that since $(u,w) \not\in RL_{parent(x)}$ and $(u,w) \in RL_{x'}$ it follows that $x'$ must be in the subtree of $T$ rooted at $x$. If we suppose, for contradiction, that there is some pair $(u', v') \in RL_{x'}$ such that $(u', v') \not\in RL_x$ then there must be a vertex $y$ on the path $P_{x, x'}$ such that $(u', v') \in RL_y$ and $(u', v') \not\in RL_{parent(y)}$. Then $y \neq x$ since $(u', v') \not\in RL_x$, and this contradicts the greedy choice property, namely, $y$ would be a topmost vertex that is located further from the root than $x$, so it should have been chosen by $GreedyRTS$.

Since $RL_{x'} \subseteq RL_x$, it follows that $S' = (S_{opt} \setminus \{x'\}) \cup \{x\}$ is another optimal solution that agrees on the first $k+1$ vertices considered by $GreedyRTS$. This contradicts the choice of $S_{opt}$ and finishes the proof of the theorem.
\end{proof}

Based on Lemma \ref{lemma:temporalPath}, the existence of a $(s,z,t)$-temporal path can be solved in polynomial time. Thus, the following theorem follows from Theorem \ref{theorem:solvabality_in_tree_base}.

\begin{theorem}
\label{thm:tree-based-main}
    The $(s,z,t)$-Temporal Separator problem is solvable in polynomial time on temporal graphs $G$ where $G_\downarrow \setminus \{s,z\}$ is a tree.
\end{theorem}


\subsection{Temporal Graphs with Bounded Pathwidth}
\label{sec:disc_reduction}

In this section, we present a reduction from the \textit{Discrete Segment Covering (DISC-SC)} problem to the $(s,z,t)$-Temporal Separator problem on graphs with bounded pathwidth. In the DISC-SC problem, we are given a set $\Gamma$ of $n$ intervals (also called segments), on the rational line and a set $\mathcal{I}$ of unit-intervals on the rational line. We wish to find a subset of unit intervals $A \subseteq \mathcal{I}$ which covers all the segments in $\Gamma$. The objective is to minimize the size of $A$. An interval $I \in \mathcal{I}$ covers a segment $S\in \Gamma$ if at least one endpoint $S$ lies in $I$. A segment $S \in \Gamma$ is covered by a set of intervals $A$ if there is an interval $I \in A$ that covers $S$. We refer to the version of DISC-SC where all segments in $\Gamma$ have length bounded by $k$ as DISC-SC-$k$. DISC-SC problem is $\mathcal{NP}$-hard~\cite{bergren2020covering}. \cite{bergren2020covering} also shows that the DISC-SC problem remains $\mathcal{NP}$-hard when the length of all segments in $\Gamma$ are equal. DISC-SC-$1$ can be solved efficiently by a simple greedy algorithm \cite{bergren2020covering}. However, the hardness of DISC-SC-$k$ for general $k > 1$ remains open.

The following theorem serves as a warm-up, and it establishes a simple polynomial time reduction from DISC-SC to the $(s,z,t)$-Temporal Separator problem. 

\begin{theorem}
\label{thm:discrete_segment_covering_red1}
    There is a polynomial-time reduction from the DISC-SC problem to the $(s,z,t)$-Temporal Separator problem.
\end{theorem}
\begin{proof}
    We denote the starting and ending points of an interval $I$ by $s(I)$ and $e(I)$, respectively. Also, we use this notation for all the segments. 
    Consider a non decreasing order $(I_1, I_2, \dots I_m)$ of all intervals in $\mathcal{I}$ (by their starting times), also consider $(C_1,C_2,\dots,C_n)$ an arbitrary order of segments in $\Gamma$. Based on the fact that the size of all intervals in $\mathcal{I}$ is one, it could be concluded that for any point $p$ and three indices $i<k<j$ if $p \in I_i$ and $p \in I_j$ the $p \in I_k$ since starting point of $I_k$ is before the starting point of $I_j$ and the ending point of $I_k$ is after the ending point of $I_i$. Now we construct a temporal graph $G = (V,E,t \times |\gamma|)$ such that $V = \{v_i| i \in [m]\}$. For any segment $C_j$ we construct the layer $G_{j\times t}$ as follows: 
    
    Let $l_s$ and $r_s$ be the indices of first and last intervals in $\mathcal{I}$ which cover the starting point $s(C_j)$. It is clear that a starting point of $C_j$ is covered by all the intervals between $I_{l_s}$ and $I_{r_s}$. Similarly, $l_e$ and $r_e$ denote the index of the first and the last intervals which cover the ending point $e(C_j)$. Since the ending point $e(C_j)$ is after the starting point $s(C_j)$ we have $l_s \leq l_e$ and $r_s \leq r_e$. So based on $l_e$ and $r_s$ we consider the following two cases:
    
    \begin{figure}
        \centering
        \includegraphics[scale=0.5,angle=0,height=3cm]{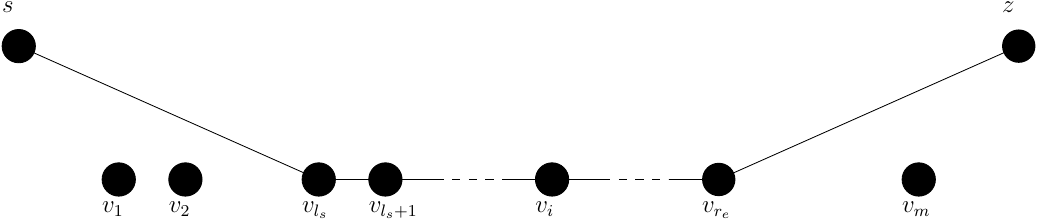}
        \caption{Demonstration of case 1 in the proof of Theorem~\ref{thm:discrete_segment_covering_red1}. Layer $G_{j\times t}$ in case that $l_e \leq r_s$. The time label for all the edges is $j \times t$}
        \label{fig:discrete_case1}
    \end{figure}
    
    \begin{figure}
        \centering
        \includegraphics[scale=0.5,angle=0,height=3cm]{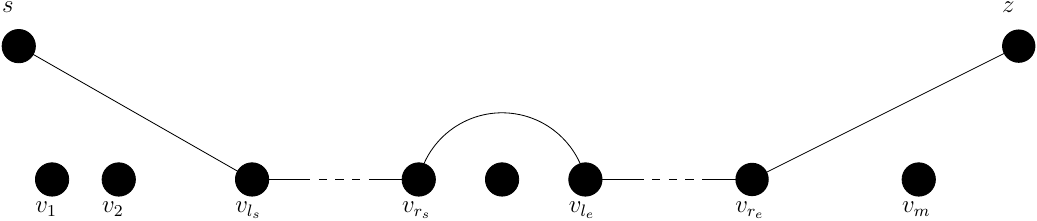}
        \caption{Demonstration of case 1 in the proof of Theorem~\ref{thm:discrete_segment_covering_red1}. Layer $G_{j\times t}$ in case that $r_s < l_e$. The time label for all the edges is $j \times t$}
        \label{fig:discrete_case2}
    \end{figure}

    \noindent\textbf{Case 1}. ($l_e \leq r_s$). In this case, we add the following temporal path which creates the layer $G_{j \times t}$. Figure \ref{fig:discrete_case1} shows this temporal path.
    \begin{gather}
        \label{eq:p_case1}
        (s,v_{l_s},j\times t), (v_{l_s},v_{l_s+1},j\times t),\dots ,
        (v_{r_e-1},v_{r_e},j\times t),(v_{r_e},z,j\times t)
    \end{gather}
    \textbf{Case 2}. ($r_s < l_e$). Similar to the previous case we add a path from $s$ to $z$ which creates the layer $G_{j\times t}$. Figure \ref{fig:discrete_case2} shows this temporal path.
    \begin{gather}
        \label{eq:p_case2}
        \nonumber(s,v_{l_s},j\times t), (v_{l_s},v_{l_s+1},j\times t),\dots ,
        (v_{r_s-1},v_{r_s},j\times t),\\
        (v_{r_s},v_{l_e},j\times t), \\
        \nonumber(v_{l_e},v_{l_e+1},j\times t),\dots,(v_{r_e-1},v_{r_e},j\times t), (v_{r_e},z,j\times t)
    \end{gather}
    
    Suppose that $A \in \mathcal{I}$, and let $S =\{v_i | I_i \in A\}$. We claim that $A$ covers $\Gamma$ if and only if  $S$ is a $(s,z,t)$-temporal separator. 
    
    $\rightarrow$ $A$ is a set of intervals that covers all segments $C_j$. If $l_e \leq r_s$ (Case 1) then there exists an interval $I_i \in A$ such that $ l_s \leq i \leq r_e$ where the temporal path which is shown in equation \ref{eq:p_case1} is incident on $v_i$. On the other hand, if $r_s < l_e$ (Case 2) then there exists an interval $I_i$ such that $l_s\leq i \leq r_s$ or  $l_e\leq i \leq r_e$  where the temporal path which is shown in equation \ref{eq:p_case2} is incident on $v_i$. Therefore, every $(s,z,t)$-temporal path in the temporal graph $G$ contains at least one vertex from $S$ resulting in $S$ being a $(s,z,t)$-temporal separator.
    
    $\leftarrow$ $S$ is a  $(s,z,t)$-temporal separator. Then for any integer $j \in [n]$ a temporal path should be incident on one vertex in $S$ in time $j \times t$. 
    If $l_e \leq r_s$ (Case 1) then there exists $v_i \in S$ such that $l_s \leq i \leq r_e$ which implies that $I_i$  covers $C_j$ and belongs to $A$. If $r_s < l_e$ then there exists $v_i \in S$ such that $l_s\leq i \leq r_s$ or  $l_e\leq i \leq r_e$ which implies that $I_i$ covers $C_j$ and belongs to $A$. Therefore all the segments are covered by an interval in $A$.
\end{proof}

The issue with the above reduction is that it does not provide any structural guarantees about the temporal graph $G$ used in the construction. In order to establish a reduction via a temporal graph $G$ whose underlying graph has bounded pathwidth, we start with a restricted version of DISC-SC, namely, the DISC-SC-$k$ problem. The following results can then be established.

\begin{theorem}
\label{thm:discrete_segment_covering_red2}
    There is a polynomial-time reduction from the DISC-SC-$k$ problem to the $(s,z,t)$-Temporal Separator in which the pathwidth of the underlying graph is bounded by $2k+6$.
\end{theorem}

\begin{proof}
    Consider an instance $(\mathcal{I}, \Gamma)$ of the Discrete Segment Covering problem such that the length of all the segments in $\Gamma$ is at most $k$. Consider intervals in $\mathcal{I} = (I_1, I_2, \dots\ I_n)$ in the non-decreasing order of their starting times. We choose a special set of intervals $SP \in \mathcal{I}$ by the following algorithm.
    \begin{enumerate}
        \item Let $SP = {I_1}$ and $index = 1$.
        \item Let $j$ be the largest index such that $s({I_j}) < e({I_{index}})$, if such $j$ exists. Otherwise, let $j = index + 1$
        \item Put $I_j$ into the set $SP$, update the integer $index$ equal to $j$ and if $j \leq n$ repeat the algorithm from step $2$.
    \end{enumerate}
    
    \begin{lemma}
    \label{lemma:covering}
        A $p$ is covered by $\mathcal{I}$ if $ SP$ covers it.
    \end{lemma}
    \begin{proof}
        We prove the lemma by induction. First, based on the algorithm, it is clear that $I_1 \in SP$ and $I_n \in SP$. Now we state the induction hyputhesis: for any $i\in [n]$ such that $I_i \in SP$, a point $p$ is covered by $\{I_1, I_2, \dots, I_i\}$ if it is covered by $\{I_1, I_2, \dots, I_i\} \cap SP$. 
        \begin{itemize}
            \item \textbf{Base case}. For $i = 1$, it is clear that $I_1 \in SP$.
            \item \textbf{Induction step}. Suppose $j<i$ is the largest integer such that $I_j \in SP \cap \{I_1, \ldots I_{i-1}\}$. Based on the inductive assumption, a point $p$ is covered by $\{I_1, I_2, \dots, I_j\}$ if it is covered by $\{I_1, I_2, \dots, I_j\} \cap SP$. Since the starting point of $I_i$ is before the ending point of $I_j$, we have that a point $p$ is covered by $\{I_1, I_2, \dots, I_i\}$ if it is covered by $\{I_1, I_2, \dots, I_i\} \cap SP$.
        \end{itemize}
        Since $I_n \in SP$, a point $p$ is covered by $\{I_1, I_2, \dots, I_n\} = \mathcal{I}$ if it is covered by $\{I_1, I_2, \dots, I_n\} $ $\cap SP = SP$.
    \end{proof}
    The main idea of the proof is based on to the following features of the special set $SP$.
    Denote $SP = \{I_{m_1},I_{m_2}, \dots I_{m_q}\}$. Based on the selection of interval $I_{m_{i+1}}$ it is clear that the starting point of $I_{m_{i+2}}$ is greater than the ending point of $I_{m_{i}}$ which implies that $s({I_{m_{i+2}}}) > s({I_{m_i}}) + 1$. More generally, we have that $e({I_{m_{i+2k}}}) > s({I_{m_i}}) + k+1$. Therefore, for any segment $C \in \Gamma$ and for any interval $I_{m_i}$ and $I_{m_j}$ such that $s(C) \in I_{m_i}$ and $e(C) \in I_{m_j}$, we could conclude that $j \leq i+2k$. This feature for $SP$ is the main idea used in constructing an instance of the $(s,z,t)$-Temporal Separator problem with low pathwidth.
    
    Now we construct a temporal graph $G = (V, E, \tau)$ where $\tau = |\Gamma|\times t$. Let $V = \{u_i| i \in [n]\} \cup \{v_i| i \in [n]\} \cup \{s,z\}$. Now, for the $i$-th segment $C \in \Gamma$ we add a path from $s$ to $z$ at time $i \times t$. Let $m_a$ and $m_b$ be the indices of the first intervals in $SP$ which cover points $s(C)$ and $e(C)$, respectively. Based on the Lemma \ref{lemma:covering} if $m_a$ (or $m_b$) does not exist, then the point $s(C)$ (respectively, $e(C)$) will not be covered by any interval in $\mathcal{I}$. Therefore, we could treat $C$ as a single point $e(C)$ (respectively, $s(C)$) and continue on with the algorithm. Let $l_s$ be the index of the leftmost interval form $\mathcal{I}$ which covers $s(C)$, and let $r_s$ be the index of the rightmost interval from $\mathcal{I}$ which covers $s(C)$. It is obvious that $s(C)$ is covered by all of the intervals between $l_s$ and $r_s$ in $\mathcal{I}$. Similarly, let $l_e$ and $r_e$ be the indices of the leftmost and the rightmost intervals which cover $e(C)$. If $l_e \leq r_s$ then consider $l_e = r_s + 1$ instead. Now, add the  following $(s,z,t)$-temporal path to the temporal graph $G$. For simplicity, we denote $i\times t$ by $\theta$.
    
    \begin{gather}
        (s,u_{l_s},\theta),(u_{l_s},v_{l_s},\theta),(v_{l_s},v_{l_s+1},\theta), \dots (v_{r_s-1},v_{r_s},\theta) \nonumber\\
        (v_{r_s},u_{r_s},\theta),(u_{r_s},u_{r_s-1},\theta), \dots (u_{m_a+1},u_{m_a},\theta) \nonumber \\
        \label{eq:path_for_discsc}(u_{m_a}, u_{m_b}, \theta) \\
        (u_{m_b},u_{m_b-1}, \theta), \dots, (u_{l_e+1}, u_{l_e}, \theta),(u_{l_e},v_{l_e},\theta) \nonumber \\
        (v_{l_e},v_{l_e}+1,\theta) \dots (v_{r_e-1}, v_{r_e},\theta), (v_{r_e},u_{r_e},\theta),(u_{r_e},z,\theta)\nonumber
    \end{gather}
    
    \begin{figure}
        \centering
        \includegraphics[scale=0.5,angle=0,height=4cm]{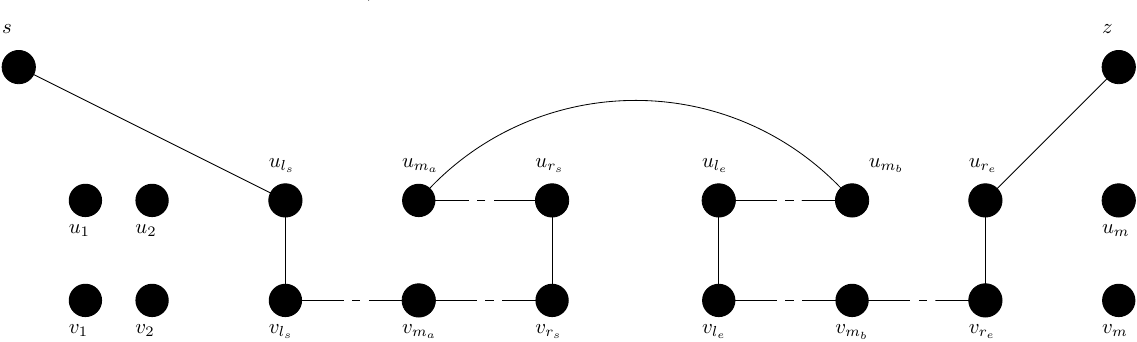}
        \caption{Demonstration of one step of the reduction in the proof of Theorem~\ref{thm:discrete_segment_covering_red2}. The figure shows the $(s,z,t)$-temporal path in the layer $G_{j\times t}$. The time label for all edges is $j \times t$.}
        \label{fig:discrete_temp_path}
    \end{figure}
    
    Figure \ref{fig:discrete_temp_path} shows the above path in the graph layer $i\times t$. We claim that there exists  $A \subseteq \mathcal{I}$ that covers $\Gamma$ with $|A| \leq p$ if and only if there is a $(s,z,t)$-temporal separator $S \subseteq V$ such that $|S| \leq p$.

    $\rightarrow$ Suppose that $A \subseteq \mathcal{I}$ covers all segments in $\Gamma$. Let $S = \{v_i | I_i \in A\}$. It is obvious that $|S| = |A|$. Now we prove that $S$ is a $(s,z,t)$-temporal separator. Suppose that there is a temporal path $P$ in $G$, based on the construction of $G$ this temporal path should be of the form shown in equation \ref{eq:path_for_discsc} for some $i \in [n]$. This implies $I_j \notin A$ for all $j$ such that $l_s < j < r_s$ or $l_e < j < r_e$ and results in the $i$-th segment not being covered by $A$. So, based on the contradiction we could conclude that $S$ is a $(s,z,t)$-temporal separator.
    
    $\leftarrow$ Suppose that $S \subseteq V$ is a $(s,z,t)$-temporal separator in a temporal graph $G$. Let $A = \{I_i| u_i \in S \text{ or } v_i \in S\}$, it is clear that $|A| \leq |S|$. Consider the $i$-th segment $C \in \Gamma$. There should be one vertex belonging to the temporal path $P$ which is shown in equation \ref{eq:path_for_discsc} in $S$ since $S$ is a $(s,z,t)$-temporal separator. Therefore there is $j$ where $l_s < j < r_s$ or $l_e < j < r_e$ and either $u_i$ or $v_i$ belong to $S$, which implies $C \in A$. Thus, $A$ covers $\Gamma$.
    
    \begin{figure}
        \centering
        \includegraphics[scale=0.5,angle=0,height=4cm]{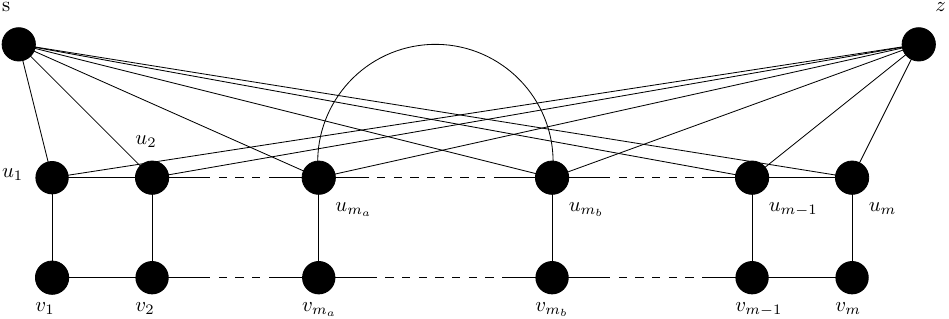}
        \caption{Illustration of the graph $G'$ which is used to show that the output of the reduction from Theorem~\ref{thm:discrete_segment_covering_red2} has bounded pathwidth. The underlying graph $G_{\downarrow}$ is a subgraph of $G'$.}
        \label{fig:discrete_under}
    \end{figure}
    
    Now we prove that the pathwidth of the underlying graph $G_{\downarrow} = (V,E')$ of the temporal graph $G(V,E,|\Gamma|\times t)$ is bounded by $2k+6$. We refer to an edge $(u_{m_a},u_{m_b},\theta)$ in a path that is shown in equation \ref{eq:path_for_discsc} as a \emph{crossing edge}. Figure \ref{fig:discrete_under} shows a graph $G'$ of which  $G_{\downarrow}$ is a subgraph. Now we give a path decomposition $(P,\beta)$ for a graph $G_{\downarrow}$ in which the width of decomposition is at most $2k+6$. Let $V(P) = \{a_1,a_2,\dots, a_m\}$ and $E(P) = \{(a_1,a_2), \dots , (a_{m-1},a(m))\}$. Let $i \in [n]$ and $l(i)$ be the largest integer such that the starting point of the interval $I_{m_{l(i)}} \in SP$ is before the starting point of interval $I_i$. Now we define the $\beta(a_i)$ as follows:
    \begin{gather*}
        \beta(a_i) = \{u_i,v_i,u_{i+1},v_{i+1},s,z\} \cup \{u_{m_l} | l \geq l(i) \text{ and } l\leq l(i)+2k\}
    \end{gather*}
    
    \begin{lemma}
    \label{lemma:path-decompose}
        For any $u_q$ and $i$, $j$, $l$ such that $i<j<l$, if $u_q \in \beta(a_i)$ and $u_q \in \beta(a_l)$, we have $u_q \in \beta(a_j)$.
    \end{lemma}
    \begin{proof}
        If $I_q \notin SP$ then it is clear that $u_q$ only appears in $\beta(a_{q-1})$ and $\beta(a_q)$. Now suppose that $I_1 \in SP$ and $q = m_p$. Since $u_{m_p} \in \beta(a_i)$ we have $m_p \leq l(i)+2k$, also $l(l) \leq m_p$ since $m_p \in \beta(a_l)$. As a result we have $m_p \leq l(i)+2k \leq l(j)+2k$ and $l(j) \leq l(l) \leq m_p$ which implies that $u_q \in \beta(a_j)$. 
    \end{proof}
    
    For any $v_i \in V$ it is clear that $v_i$ just belongs to the two sets $\beta(a_{i-1})$ and $\beta(a_i)$. Also, $s$ and $z$ are present in all the sets. Therefore, by Lemma \ref{lemma:temporalPath} we could say that the third property of path decomposition is satisfied. So, it is sufficient to show that for every edge $(u,v) \in E(G_{\downarrow})$ there exists $i \in [n]$ such that $\{u,v\} \subseteq \beta(a_i)$. If the edges are not crossing edges, then there are three types of edges $(u_i,v_i)$, $(u_i,u_{i+1})$, and $(v_i,v_{i+1})$ which satisfy the condition by the definition of $\beta(a_i)$.
    If $e = (u_i, u_j)$ is a crossing edge, then $I_i \in SP$ and $I_j\in SP$, so let $m_p = i$ and $m_q = j$. Since this edge corresponds to a segment $C$ such that $s(C) \in I_{m_p}$ and $e(C) \in I_{m_q}$ we could conclude that $m_q \leq m_p + 2k$ which implies that ${u_i, u_j} \subseteq \beta(a_i)$.
    
    Also, the cardinality of all sets $\beta(a_i)$ is $2k+7$ which implies that the width of $(P,\beta)$ is  $2k+6$. Therefore the pathwidth of the underlying graph $G_{\downarrow}$ is at most $2k+6$.
\end{proof}

The significance of this result is that if one hopes to design efficient algorithms for the $(s,z,t)$-Temporal Separator with bounded pathwidth one is faced with an obstacle of resolving the hardness of DISC-SC-$k$ problem, as stated in the following theorem.
\begin{theorem}
\label{thm:pathwidth-main}
    If the $(s,z,t)$-Temporal Separator problem on temporal graphs with bounded pathwidth is solvable in polynomial time then the DISC-SC-$k$ problem is solvable in polynomial time.
\end{theorem}


\section{Conclusions}
\label{sec:conclusions}
In this work, we defined the $(s,z,t)$-Temporal Separator problem, generalizing the $(s,z)$-Temporal Separator problem. We showed that $(s,z)$-Temporal Separator and $(s,z,t)$-Temporal Separator problems could be approximated within $\tau$ and $\tau^2$ approximation ratio, respectively, in a graph with lifetime $\tau$. We also presented a lower bound $\Omega(\log(n) + \log(\tau))$ for polynomial time approximability of $(s,z,t)$-Temporal Separator assuming that $\mathcal{NP}\not\subset\mbox{\sc Dtime}(n^{\log\log n})$. Then we considered special classes of graphs. We presented two efficient algorithms: one for temporal graphs $G$ with $bw(G_\downarrow) \le 2$ and one for temporal graphs $G$ with $G_\downarrow \setminus \{s,z\}$ being a tree. The question of whether there is a polynomial-time algorithm to compute a minimum $(s,z,t)$-temporal separator in a temporal graph of bounded treewidth remains an interesting open problem. However, we showed a reduction from the DISC-SC-$k$ problem to $(s,z,t)$-Temporal Separator when the pathwidth of the underlying graph is bounded by a constant number. Therefore, designing efficient algorithms for bounded treewidth graphs encounters serious obstacles, such as making progress on the open problem of the hardness of DISC-SC-$k$. Another interesting direction of future research is to consider temporal separator problems with the additional restriction of ``balancedness'', as discussed at the end of Section~\ref{sec:related}.

\bibliographystyle{plainurl}
\bibliography{bibfile}

\end{document}